\documentclass[a4paper,10pt]{article}
\usepackage{amsthm}
\usepackage{amsmath}
\usepackage{amssymb}
\usepackage{latexsym}
\usepackage{xcolor}
\usepackage{hyperref}

\newtheorem{theorem}{Theorem}[section]

\newtheorem{lemma}[theorem]{Lemma}

\newtheorem{remark}[theorem]{Remark}
\newtheorem{definition}[theorem]{Definition}
\newtheorem{example}[theorem]{Example}
\newtheorem{assumption}[theorem]{Assumption}

\begin{document}

\title{Robust utility maximization in markets with transaction 
costs\thanks{Both authors were supported by %
the ``Lend\"{u}let'' grant LP2015-6 of the Hungarian Academy of Sciences and by the NKFIH 
(National Research, Development and Innovation Office, Hungary) 
grant KH 126505. The authors thank Walter Schachermayer and two anonymous referees for helpful
comments.}}
\author{Huy N. Chau \and Mikl\'os R\'asonyi}
\date{\today}
\maketitle

\begin{abstract} We consider a continuous-time market with proportional transaction costs. Under appropriate
assumptions we prove the existence of optimal strategies for investors who maximize their worst-case utility over
a class of possible models. We consider utility functions defined either on the positive axis or on the 
whole real line. 
\end{abstract}

\section{Introduction}

In this paper, the existence of solutions to the utility maximization problem 
from terminal utility is studied in the presence of model ambiguity. 
We assume that investors prepare for the worst-case scenario in the sense that they take the infimum of utility functionals 
over the class of possible models before maximizing over admissible investment strategies.

The literature on robust optimization typically assumes that uncertainty is modeled by a family of prior measures 
$\mathcal{P}$ on some canonical space in which trajectories of the processes lie. Starting with \cite{quenez}, \cite{schied06}, the case in which $\mathcal{P}$ is dominated by a reference measure $P_*$ 
has received ample treatment. In diffusion settings this corresponds to uncertainty in the drift. 
Such an approach is not completely convincing since market participants
may also be uncertain about the volatilities. 

More recently, the non-dominated problem has also been studied in various contexts. 
For instance, \cite{tevzadze} investigated a compact set of possible drift and volatility coefficients and tackled 
the robust problem by solving an associated Hamilton--Jacobi--Bellman equation. 
In \cite{mdc}, where volatility coefficients are uncertain over a compact set and the drift is known, 
the theory of BSDEs is applied.
Existence results in a fairly general class of models are available
only in discrete time: see \cite{nutz}, \cite{blanchard-carassus2017},
\cite{neufeld-sikic2017}, \cite{bartl2017}, \cite{kupper} and \cite{andrea}.
A minimax result was established for bounded utilities in frictionless
continous time markets in \cite{denis-kervarec}. 

As far as we know, our existence results below are the first to apply in a
broad class of continuous-time models. We now summarize the principal
ideas underlying our arguments. First, we work under
proportional transaction costs. In this setting strategies can be identified with
finite variation processes which we endow with a suitable topology. Second,
instead of a family of measures we consider a parametrized family of stochastic processes
on a fixed filtered probability space. Necessarily, instead of one portfolio value 
we need to consider a family of possible values corresponding to the respective parameters.
Third, the latter fact forces us to take the family of strategies as our
domain of optimization (unlike most of the optimal investment literature since \cite{ks99} which
prefer to optimize over a set of random variables: the terminal values of possible portfolios).
Fourth, we exploit that an appropriate boundedness of terminal portfolio values implies
appropriate boundedness of the strategies themselves: this is false in continuous-time frictionless markets
but true in our setting. Fifth, we profit from a method first developed in \cite{m} that verifies
the supermartingale property of a putative optimizer, based on a lemma of \cite{do}.
Because of the fourth point above our techniques do not
seem to be applicable in the continuous-time frictionless setting. Note however the companion paper \cite{andrea} 
which treats \emph{discrete-time} frictionless markets.

The robust model in this paper, similar to the ones introduced in \cite{bp}, \cite{nutz2}, \cite{lr}, 
assumes that there is a parametrization for the uncertain dynamics of risky assets. However, as we will see
below, no specific assumption is
made about the parametrization and an arbitrary index set is permitted.
From a practical point of view, this approach is particularly tractable and easily implemented when it comes to calibration. 
For example, estimating drift and volatility parameters for diffusion price processes, the results only give guesses (hopefully
with some confidence sets) about the true values. 
Thus it is reasonable to parametrize ambiguity by considering suitable ranges which contain possible values for 
the coefficients being estimated. 

From a mathematical point of view, the treatment of robust models in the present paper 
simplifies technical issues, as it will become apparent from the proofs. 
Working on the same (filtered) probability space, instead of considering a family of measures, 
gives us more flexibility by avoiding the canonical setting with problems concerning null events, filtration completion, etc. 
Measurable selection arguments, see \cite{bn15}, \cite{biagini2015robust} or \cite{nutz}, are not needed anymore. 
Our approach can still
incorporate most of the relevant models classes and their laws do not need to be equivalent, see Section
\ref{sec:model}.


Compactness plays an important role in proving the existence of optimizers. Usually, the utility maximization problem is 
transformed into an ``abstract'' version with random variables (the terminal wealth of admissible portfolios), and then convex 
compactness results in $L^0$, in particular, 
Koml\'os-type arguments, are applied successfully, see \cite{ks99}. Unfortunately, the robust setting 
is unlikely to be lifted 
to ``abstract'' versions, since the uncertainty produces a whole collection of wealth processes. 
As a result, Koml\'os-type arguments on the space $L^0$ cannot be employed. Furthermore, the candidate 
dual problem in this setting does not, in general, admit a solution 
(see Remark 2.3 of \cite{bartl2017}) so the usual approach of getting optimizers 
from solutions of dual problems seems inapplicable. Therefore, we are forced to work on the primal problem directly. 

We are using two Koml\'os-type arguments: the first one is performed on the space of finite variation processes (strategies), 
which gives a candidate for the optimizer and the second is used in an Orlicz space context, to handle possible losses of trading when
establishing the supermartingale property of the optimal wealth process, relying on \cite{do}. 
A crucial observation is that the utility of a portfolio is a sequentially upper
semicontinuous function of the strategies (when the latter are equipped
with a convenient convergence structure), see \cite{paolo} where the optimization
problem was viewed in a similar manner.

The paper is organized as follows. Section \ref{sec:model} introduces the robust market model and technical assumptions. 
Sections \ref{sec:U_pos} and \ref{sec:U_R} study the existence of solutions of the robust utility maximization problems when 
the utility functions are defined on $\mathbb{R}_+$ and $\mathbb{R}$, respectively. Ramifications are discussed in Section
\ref{qclu}. Some preliminaries on  
finite variation processes and on Orlicz space theory are presented in Section \ref{appendix}. 

\section{The market model}\label{sec:model}

Let $(\Omega, \mathcal{F}, (\mathcal{F}_t)_{t\in [0,T]}, P)$ be a filtered probability space, 
where the filtration is assumed to be right-continuous and $\mathcal{F}_0$ coincides with the $P$-completion of the trivial sigma-algebra. We denote the class of real-valued random variables by $L^0$ and its
positive cone by $L^0_+$.

Let $\Theta$ be a (non-empty) set, which is interpreted as the parametrization of uncertainty. We consider a financial market consisting of a riskless asset $S^0_t = 1$ for all $t \in [0,T]$ and a risky asset, whose dynamics is unknown.
To describe the latter, we consider a family $(S^{\theta}_t)_{t\in [0,T]}$, 
$\theta \in \Theta$ of adapted, positive processes with continuous trajectories
which represent the possible price evolutions. 
No condition is imposed on $\Theta$ and, for the moment, on the dynamics of the risky asset either. 

\begin{remark} {\rm We now comment on the difference between our concept of model ambiguity and
that of most previous papers, where a family of priors is considered on a canonical space.

Working on a given probability space and filtration amounts to fixing the information structure
of the problem: the information flow is normally generated by a particular diriving process
(such as a multidimensional Brownian motion). Possible prices are then functionals of a parameter (finite or
infinite dimensional, see Examples \ref{gesa} and \ref{upo} below) and the driving noise.
Strategies are functionals adapted to the given information flow.

Considering a family of probabilities, one has greater liberty in the sense that no common driving 
noise is required, but
the choice of strategies is limited: they must be adapted functionals on
the canonical space.  In a sense, they must be ``closed loop'' controls depending on the price process.
In our modelling the controls are ``open loop'', they are adapted to an information flow that may be strictly
bigger than the natural filtration of any possible price process.

In a strictly formal sense none of two the approaches is more general than the other, 
see also examples in \cite{andrea}. Intuitively, the standard setting is the more general one,
while ours seems more easily tractable and it fits better a practical calibration and/or statistical 
inference framework.}
\end{remark}

We illustrate, by the following examples, that the present setting is useful and contains interesting models from 
previous studies.

\begin{example}\label{gesa} {\rm (The robust Black-Scholes market model.) The risky asset satisfies
the SDE
$$dS^{(\mu, \sigma)}_t = S^{(\mu, \sigma)}_t(\mu dt + \sigma dW_t),\ S_0^{(\mu,\sigma)}=s_0>0$$
where $\mu, \sigma$ are constants and $W$ is a standard Brownian motion. The uncertainty is modeled by
$$\Theta = \{ \theta=(\mu, \sigma) \in \mathbb{R}^2: \underline{\mu} \le \mu \le \overline{\mu},\ \underline{\sigma} \le \sigma \le \overline{\sigma} \},$$
where $\underline{\mu}\leq\overline{\mu}$, $0<\underline{\sigma}\leq\overline{\sigma}$ are given constants.
The classical Black-Scholes model corresponds to the case $\underline{\mu} = \overline{\mu}$ and $\underline{\sigma} = \overline{\sigma}$. It is observed that the laws of $S^{\mu_1,\sigma_1}, S^{\mu_2,\sigma_2}$ are singular when $\sigma_1 \ne \sigma_2$. If only volatility uncertainty is considered, then the family of laws is mutually singular.  See \cite{lr} and \cite{bp} about treatments for similar models. }
\end{example}

\begin{remark} {\rm In the domain of robust finance, measurable selection techniques
are often used, see e.g. \cite{nutz}. This requires certain measurability
of the family of laws corresponding to various models. In our present approach, however,
this is not a necessity. Let e.g. $\Theta'$ be a non-Borelian (or even
non-analytic) subset of $\Theta$ in Example \ref{gesa} above. Theorems \ref{thm_1} and \ref{thm9}
apply to the family of models $S^{\theta}$, $\theta\in \Theta'$, too.}
\end{remark}

\begin{example}\label{upo}{\rm In the above example, $\Theta$ was a subset of a finite-dimensional Euclidean space.
One may easily fabricate similar examples where $\Theta$ is infinite-dimensional. For instance,
let $\Theta$ consist of all pairs of predictable processes $(\mu_t,\sigma_t)$ such that, for all $t\in [0,T]$,
$\mu_t\in [\underline{\mu},\overline{\mu}]$ a.s. and $\sigma_t\in [\underline{\sigma},\overline{\sigma}]$
a.s. and consider the SDEs
$$dS^{(\mu, \sigma)}_t = S^{(\mu, \sigma)}_t(\mu_t dt + \sigma_t dW_t),\ S_0^{(\mu,\sigma)}=s_0>0,$$
for each $(\mu,\sigma)\in\Theta$.}
\end{example}

The following example extends the robust Black-Scholes model and allows an external economic factor.

\begin{example} {\rm (A factor model inspired by \cite{t12}, but much
simplified.) 
Let $\Theta \subset \mathbb{R}^{2\times 2}$ be a set. 

The risky asset is governed by the SDE
$$
dS^{\theta}_t = S^{\theta}_t(m(Y^{\theta}_t) + \sigma(\theta^{11}Y^{\theta}_t + \theta^{21}) )dt + \sigma dW^{1}_t),\
S^{\theta}_0=s_0>0, $$
and the factor process evolves according to
$$dY^{\theta}_t = \left[g(Y^{\theta}_t) + \left\langle \rho, \theta^{1\cdot}Y^{\theta}_t + \theta^{2\cdot} \right\rangle \right]  dt + \rho_1 dW^{1}_t + \rho_2dW^{2}_t,
\ Y_0^{\theta}=y_0,$$
where $m,g$ are suitable functions, $W=(W^1,W^2)$ is a two dimensional  Brownian motion and $\rho=(\rho_1,\rho_2)\in\mathbb{R}^2$ are correlation parameters. The bracket $\langle\cdot,\cdot\rangle$ denotes
scalar product in $\mathbb{R}^2$. Note that the original
setting of \cite{t12} cannot be directly transferred to the
present one as it involves a family of weak solutions of SDEs which are
not necessarily definable on our given stochastic basis.} 
\end{example}

The risky asset is traded under proportional transaction costs $\lambda \in (0,1)$. More precisely, investors have to pay a higher (ask) price $S^{\theta}$ when buying the risky asset but receive a lower (bid) price $(1-\lambda)S^{\theta}$ when selling it.

Let $\mathcal{V}$ denote the family of non-decreasing, right-continuous functions on $[0,T]$ which are $0$ at time $0$.
Let $\mathbf{V}$ denote the set of triplets 
$H=(H^{\uparrow},H^{\downarrow},H_0)$ where 
$H^{\uparrow}_t,H^{\downarrow}_t$, $t\in [0,T]$ are optional processes 
such that $H^{\uparrow}(\omega),H^{\downarrow}(\omega)\in\mathcal{V}$ for each 
$\omega\in\Omega$ and $H_0\in\mathbb{R}$ (deterministic). The space $\mathbf{V}$ can be equipped with a convergence structure, see Subsection \ref{section:v_space} below for details.

Each trading strategy corresponds to an element $H\in\mathbf{V}$. 
In this formulation, $H^{\uparrow}$ denotes the cumulative amount of transfers from the riskless asset to the risky one and 
$H^{\downarrow}$ represents the transfers in the opposite direction,
$H_0$ encodes the amount of initial transfer from the riskless
asset to the risky one. 
Therefore the portfolio position in
the risky asset at time $t$ equals $\phi_t:=H_0+H^{\uparrow}_t - H^{\downarrow}_t$, $t\in [0,T]$,
$\phi_{0-}:=0$.

For any real number
$x\in\mathbb{R}$, we denote $x^+:=\max\{0,x\}$, $x^-:=\max\{0,-x\}$.
For an initial capital $x\in\mathbb{R}$, the dynamics of cash account of an investor following strategy $H$ evolves according to
$$W^x_t(\theta,H) := x - H_0^+S^{\theta}_0+ H_0^-S^{\theta}_0(1-\lambda)- \int_0^t{S^{\theta}_u dH^{\uparrow}_u} + 
\int_0^t{(1-\lambda)S^{\theta}_udH^{\downarrow}_u},$$
for $t\in [0,T]$.

The liquidation value is defined by 
\begin{eqnarray}\label{wl}
W^{x,liq}_t(\theta, H) &:=& W^x_t(\theta,H) + \phi^{+}_t(1-\lambda)S^{\theta}_t - \phi^-_tS^{\theta}_t. 
\end{eqnarray}

We introduce the definition of consistent price systems, which play a similar role to martingale measures in frictionless markets, see \cite{ks}, \cite{grw}, and \cite{grw2}.

\begin{definition}
For each $\theta \in \Theta$, a $\lambda$-\emph{consistent price system} ($\lambda$-CPS) for the
model ${\theta}$ is a pair $(\tilde{S}^{\theta}, Q^{\theta})$ of a probability measure $Q^{\theta}\sim P$ and a $Q^{\theta}$ local martingale $\tilde{S}^{\theta}$ such that 
\begin{equation}\label{matroz}
(1- \lambda)S^{\theta}_t \le \tilde{S}^{\theta}_t \le S^{\theta}_t, \qquad a.s. \qquad \forall t \in [0,T].
\end{equation}
A $\lambda$-\emph{strictly consistent price system} ($\lambda$-SCPS) is a CPS such that
the inequalities are strict in \eqref{matroz}.
\end{definition}

We will impose the existence of consistent price systems for every model
$S^{\theta}$. In Section \ref{sec:U_pos}, we will need the following assumption,
in order to be able to use the results of \cite{csy}.

\begin{assumption}\label{strong}
	For each $\theta \in \Theta$ and for all $0<\mu<\lambda$, the price process $S^{\theta}$ admits 
	a $\mu$-CPS.
\end{assumption}
{This assumption is fulfilled iff, for every $\theta \in \Theta$, the process $S^{\theta}$ satisfies the no arbitrage condition for $\mu$-transaction cost for all $\mu>0$, see \cite{grw}. 
See Example \ref{example_cost} for a risky asset violating Assumption \ref{strong}.}

Clearly, a $\mu$-CPS is also a $\lambda$-SCPS. 

\begin{lemma}
If $(\tilde{S}^{\theta}, Q^{\theta})$ is a $\lambda$-strictly consistent price system, then the random variable 
\begin{equation}\label{ep_strict}
\delta(\theta):=\inf_{t \in [0,T]} \min\{\tilde{S}^{\theta}_t - (1-\lambda)S^{\theta}_t, S^{\theta}_t - \tilde{S}^{\theta}_t\}
\end{equation}
is almost surely strictly positive.
\end{lemma}
\begin{proof}
The argument follows that of Lemma 3.6.4 in \cite{ks}. 
\end{proof}

Let $$
\mathcal{M}^{\theta}:=\{dQ^{\theta}/dP:\ (\tilde{S}^{\theta},Q^{\theta})\mbox{ is a }\lambda\mbox{-CPS}\}.
$$ 

For a consistent price system $(\tilde{S}^{\theta}, Q^{\theta})$, we define the process
\begin{equation}\label{process_v}
V_t^x(\theta,H) := W^x_t(\theta,H) + \phi_t \tilde{S}^{\theta}_t, 
\end{equation}
without emphasizing the dependence of $V$ on the specific consistent price system. It is easy to check that $W^{x,liq}_t(H) \le V^x_t(\theta,H)$ a.s.,
for all $t \in [0,T]$.

\section{Utility function on $\mathbb{R}_+$}\label{sec:U_pos}

\begin{assumption}\label{U_positive}
The utility function $U :(0,\infty) \to \mathbb{R}$ is nondecreasing and concave. 
\end{assumption} 

Define the convex conjugate of $U$ by
$$V(y):=\sup_{x >0}(U(x)-xy), \qquad y>0.$$
Admissible strategies are defined in a natural way, thanks to the domain of the utility function.

\begin{definition}
A strategy $H=(H^{\uparrow},H^{\downarrow},H_0)\in\mathbf{V}$ is \emph{admissible} for initial capital $x>0$ 
and for the model $\theta\in\Theta$, if, for all $t \in [0,T]$, 
$$
W^{x,liq}_t(\theta, H) \ge 0\mbox{ a.s.}
$$
Denote by $\mathcal{A}^{\theta}(x)$ the set of all admissible strategies for $\theta$.
Set 
$$
\mathcal{A}^{\theta}_0(x):=\{H\in\mathcal{A}^{\theta}(x):\ \phi_T=H_0+H^{\uparrow}_T-H^{\downarrow}_T=0\},
$$
and
$\mathcal{A}(x) = \bigcap_{\theta \in \Theta} \mathcal{A}_0^{\theta}(x)$. 
\end{definition}

\begin{remark} {\rm For each 
$H\in\mathcal{A}(x)$, 
$W^{x,liq}_T(\theta,H) = W^x_T(\theta,H)=V^x_T(\theta,H)$
by 
$\phi_T=0$. 
We also see from (\ref{wl}) that at time $0 < t < T$, the liquidation value is neither 
concave nor convex in $H$. However, the condition $\phi_T = 0$ recovers concavity of 
the liquidation value with respect to $H$ at time $T$. This is crucial for finding 
maximizers in the subsequent analysis.}
\end{remark}

Let $x>0$. Note that $\mathcal{A}(x)\neq\emptyset$ since the identically zero strategy is therein.
Investors want to find the optimizer for
\begin{equation}\label{prob_con}
u(x): = \sup_{H\in\mathcal{A}(x)} \inf_{\theta \in \Theta} EU(W^{x,liq}_T(\theta,H)).
\end{equation}
It is worth noting that maximizing in $H$ is a concave problem, however, 
minimizing over $\Theta$ is \emph{not} a convex problem. 

For each $\theta \in \Theta$ and $x>0$, we denote 
$$\mathcal{C}^{\theta}(x) := 
\left\lbrace  X\in L^0_+:\, X\leq W^{x,liq}_T(\theta,H)\mbox{ for some }H \in \mathcal{A}^{\theta}(x)\right\rbrace.$$
For each $y>0$, the set of \emph{supermartingale deflators} $\mathcal{B}^{\theta}(y)$ consists of the strictly positive processes $Y=(Y^0_t,Y^1_t)_{t\in [0,T]}, Y^0_0 =y$ such that $\frac{Y^1}{Y^0} \in [(1-\lambda)S^{\theta},S^{\theta}]$ and $W^{x}(\theta,H) Y^0 + \phi Y^1$ is a (c\`adl\`ag) 
supermartingale for all $H \in \mathcal{A}^{\theta}(x)$.
Also, we define 
$$\mathcal{D}^{\theta}(y):= \{ Y^0_T: (Y^0,Y^1) \in \mathcal{B}(y) \}.$$ 
The primal and dual value functions for the $\theta$-model are
$$u^{\theta}(x):= \sup_{f \in \mathcal{C}^{\theta}(x)} EU(f), \qquad v^{\theta}(y) := \inf_{h \in \mathcal{D}^{\theta}(y)} E V(h).$$

The next lemma states that the sets $\mathcal{C}^{\theta}(x)$ and $\mathcal{D}^{\theta}(y)$ are polar to each 
other. It follows directly from Proposition 2.9 of \cite{csy}. 

\begin{lemma}\label{q} Fix $x,y>0$. Let Assumption \ref{strong} be in force.
A random variable $X\in L^0_+$ satisfies $X\in\mathcal{C}^{\theta}(x)$
iff $EXY\leq xy$ for all $Y\in\mathcal{D}^{\theta}(y)$. A random variable $Y\in L_+^0$ satisfies
$Y\in\mathcal{D}(y)$ iff $EXY\leq xy$ for all $X\in\mathcal{C}^{\theta}(x)$.\hfill $\Box$
\end{lemma}

We impose a technical assumption.
\begin{assumption}\label{finite}
The dual value function $v^{\theta}(y),y>0$ is finite for \emph{all} $\theta \in \Theta$.
\end{assumption}

\begin{theorem}\label{thm_1}
Let $x>0$. Under Assumptions \ref{strong}, \ref{U_positive}, \ref{finite}, the robust utility maximization problem 
(\ref{prob_con}) admits a solution, i.e. there is $H^*\in\mathcal{A}(x)$ satisfying 
\begin{equation*}
u(x) = \inf_{\theta \in \Theta} EU(W^{x,liq}_T(\theta, H^*)).  
\end{equation*}
When $U$ is bounded from above, the same conclusion holds assuming only that there exists (at least)
one $\tilde{\theta}\in\Theta$ for which there exists a $\lambda$-SCPS.  
\end{theorem}
\begin{proof} If $U$ is constant then there is nothing to prove. Otherwise, by adding a constant
to $U$, we may assume that $U(\infty)>0>U(0)$.

Notice that $U(\infty) > 0$ and
\begin{equation}\label{bante}
u^{\theta}(x)\geq U(x)
\end{equation}
imply $\liminf_{x\to\infty}u^{\theta}(x)/x\geq 0$. From Lemma \ref{q}, trivially,
\begin{equation}\label{tante}
u^{\theta}(x) \le v^{\theta}(y) + xy,
\end{equation}
for all $y>0$. Fixing $y$, we obtain $\limsup_{x\to\infty}u^{\theta}(x)/x \le y$ and 
sending $y$ to zero gives 
\begin{equation}\label{inada}
\lim_{x\to \infty} \frac{u^{\theta}(x)}{x} = 0.
\end{equation}

After these preparations, we turn to the main arguments.
Assumption \ref{finite}, \eqref{tante} and \eqref{bante} imply that $u^{\theta}(x)$ is finite for each 
$\theta$ and so is $u(x)$. 
Let $H^n \in \mathcal{A}(x)$, $n\in\mathbb{N}$ be a maximizing sequence, i.e.
$$\inf_{\theta \in \Theta} EU(W^{x,liq}_T(\theta, H^n)) \uparrow u(x),\ n\to\infty.$$

Let us fix, for the moment, $\theta\in\Theta$ 
and a $\mu$-CPS $(\tilde{S}^{\theta}, Q^{\theta})$ with $0<\mu < \lambda$. First, we prove that the process 
\begin{equation}\label{s}
V^x_t(\theta,H^n) = W^x_t(\theta,H^n) + \phi^n_t \tilde{S}^{\theta}_t,
\end{equation}
is a $Q^{\theta}$-supermartingale for all $n$. Indeed, It\^o's formula gives
\begin{eqnarray*}
dV^x_t(\theta,H^n) &=& - S^{\theta}_tdH^{n,\uparrow}_t + (1-\lambda)S^{\theta}_tdH^{n,\downarrow}_t + \tilde{S}^{\theta}_td\phi^n_t + \phi^n_td\tilde{S}^{\theta}_t\\
&=& (\tilde{S}^{\theta}_t - S^{\theta}_t)dH^{n,\uparrow}_t + [(1-\lambda)S^{\theta}_t - \tilde{S}^{\theta}_t]dH^{n,\downarrow}_t + \phi_t^nd\tilde{S}^{\theta}_t.
\end{eqnarray*}
Admissibility of $H^n$ implies  
\begin{eqnarray}\label{bla}
& & H^{n,+}_0 (S^{\theta}_0-\tilde{S}^{\theta}_0) + \int_0^t (S^{\theta}_u - \tilde{S}^{\theta}_u)
dH^{n,\uparrow}_u \\
\nonumber &+& \int_0^t [\tilde{S}^{\theta}_u - (1-\lambda)S^{\theta}_u]dH^{n,\downarrow}_u  + \left( 
\int_{0}^t{\phi^n_u d\tilde{S}^{\theta}_u} \right)^- \\
\nonumber &\le& x + H^{n,-}_0 [S^{\theta}_0(1-\lambda)-\tilde{S}_0^{\theta} ]+  
\left( \int_{0}^t{\phi^n_u\, d\tilde{S}^{\theta}_u} \right)^+ .
\end{eqnarray}
In particular, we obtain $\left( \int_0^t{\phi^n_u d\tilde{S}^{\theta}_u} \right)^- \le x + H^{n,-}_0 S^{\theta}_0(1-\lambda) $ 
for every $t \in [0,T]$ and therefore $\int_0^t{\phi^n_u d\tilde{S}^{\theta}_u}$, $t\in [0,T]$ is a 
$Q^{\theta}$-supermatingale, see \cite{as}. It follows that $V_t^x(\theta,H^n)$, $t\in [0,T]$ is also
a $Q^{\theta}$-supermartingale.  

We claim that $\sup_n H_0^{n,-}$ is finite. If this were not the case then, along a subsequence $n_k$,
$k\in\mathbb{N}$ we would have $H_0^{n_k,-}\to\infty$, $k\to\infty$ and $H^{n_k,+}_0=0$, $k\in\mathbb{N}$. Taking $Q^{\theta}$-expectation
in \eqref{bla} we would get 
$$
0\leq x+\liminf_{n\to\infty}  H^{n,-}_0(S_0^{\theta}(1-\lambda)-\tilde{S}_0^{\theta})=-\infty,
$$ 
a contradiction. Hence the supremum is indeed finite.

Furthermore, from the supermartingale property of $\int_0^t{\phi^n_u d\tilde{S}^{\theta}_u}$,
$$\sup_n E^{Q^{\theta}}\left( \int_0^T{\phi^n_u d\tilde{S}^{\theta}_u} \right)^+ \le x + \sup_nH^{n,-}_0 S^{\theta}_0(1-\lambda) $$
follows.
Using (\ref{ep_strict}), we deduce that
\begin{eqnarray*}
\sup_n E^{Q^{\theta}}  \int_0^T{ \delta (\theta) \left( dH^{n,\uparrow}_u + dH^{n,\downarrow}_u\right)}  & \le&  \\
\sup_n E^{Q^{\theta}} \left( \int_0^T{ (S^{\theta}_u - \tilde{S}^{\theta}_u)dH^{n,\uparrow}_u + [\tilde{S}^{\theta}_u - (1-\lambda)S^{\theta}_u ]dH^{n,\downarrow}_u} \right) &<& \infty. 
\end{eqnarray*}
Lemma \ref{gom} implies that there exist convex weights $\alpha^n_j \ge 0$, $j=n,...,M(n)$, and 
$\sum_{j=n}^{M(n)} \alpha_j^n = 1$ such that 
$\tilde{H}^n:= \sum_{j=n}^{M(n)} H^n \to H^*$ in $\mathbf{V}$. Since convex combinations improve 
utility of concave functions, we obtain that $\tilde{H}^n, n \in \mathbb{N}$ is also a maximizing sequence, 
$$ \inf_{\theta \in \Theta} EU(W^{x,liq}_T(\theta, H^n)) \le \inf_{\theta \in \Theta} EU(W^{x,liq}_T(\theta, \tilde{H}^n)) \to 
u(x),\ n\to\infty.
$$

We now prove that the sequence  $U^+(W^{x,liq}_T(\theta,\tilde{H}^n)), n \in \mathbb{N}$ is uniformly integrable for 
each $\theta\in\Theta$. Suppose, by contradiction, that the sequence is not uniformly integrable
for some $\theta$. Then one could find disjoint sets $A_n\in\mathcal{F}$, $n\in\mathbb{N}$ and a 
constant $\alpha > 0$ such that
$$E\left(  U^+(W^{x,liq}_T(\theta,\tilde{H}^n))1_{A_n} \right) \ge \alpha, \qquad \text{ for } n \ge 1.$$
Set $w^n = \sum_{i=1}^n W^{x,liq}_T(\theta,\tilde{H}^i) 1_{W^{x,liq}_T(\theta,\tilde{H}^i) \ge u_0}1_{A_i}$, 
where $u_0$ is chosen such that it 
satisfies $U(u_0) = 0$. It is immediate that 
$$EU(w^n) = \sum_{i=1}^n E\left(  U^+(W^{x,liq}_T(\theta,\tilde{H}^i))1_{A_i} \right) \ge n\alpha.$$
In addition, for any $h \in \mathcal{D}^{\theta}(1)$, the supermartingale property shows that 
$E[hw^n] \le nx$. Consequently, we obtain $w^n \in C^{\theta}(nx)$, by Lemma \ref{q}. We compute
$$ \frac{u^{\theta}(nx)}{nx} \ge  \frac{EU(w^n)}{nx} \ge \frac{\alpha}{x} > 0$$
and passing to the limit when $n\to\infty$ contradicts (\ref{inada}). 
Thus, $U^+(W^{x,liq}_T(\theta,\tilde{H}^n))$ $n\in\mathbb{N}$ is indeed uniformly integrable. 

Since $\tilde{H}^n \to H^* \in \mathbf{V}$, $W^{x,liq}_T(\theta, \tilde{H}^n) \to W^{x,liq}_T(\theta, H^*)$ almost surely by Remark \ref{as}, so Fatou's lemma and uniform integrability imply 
\begin{align*}
\limsup_{n\to\infty}\left(\inf_{\theta \in \Theta} EU(W^{x,liq}_T(\theta, \tilde{H}^n)) \right) &\le \inf_{\theta \in \Theta} \limsup_{n\to\infty} EU(W^{x,liq}_T(\theta, \tilde{H}^n))\\
&\le  \inf_{\theta \in \Theta} E U(W^{x,liq}_T(\theta, H^*))  
\end{align*}
which proves that $H^*$ is the optimizer. It remains to check that $H^*\in\mathcal{A}(x)$. For each $\theta$, $W^{x,liq}_t(\theta,H^*)\geq 0$ a.s.
for Lebesgue-almost every $t$, by Remark \ref{as} so
we get admissibility of $H^*$ since $t\to W^{x,liq}_t$ is a.s. right-continuous.

In the case where $U$ is bounded from above, it is enough to
perform the first part of the above proof for $\tilde{\theta}$ 
and then, for \emph{each} $\theta$, one may simply invoke Fatou's lemma to complete the proof.
\end{proof}

\begin{remark} {\rm In the classical theory where there is no uncertainty, i.e. 
when $\Theta$ contains only one element, the existence result holds assuming the finiteness of $u(x)$ only. 
This condition, however, does not warrant
to find optimizers in the robust problem. Indeed, the finiteness of $u(x)$ makes the 
robust problem well-posed, compactness gives a candidate for the optimizer, 
but this is still not enough to prove that the candidate is indeed the optimizer. 
To complete the proof, it is necessary to have upper-semicontinuity of the expected 
utility when considered as a function of the strategy variable. In \cite{nutz}, a counterexample 
(in which $u(x)$ is finite but one could not find the optimizer) is given in the nondominated case. 
The author's argument exploits precisely the lack of upper-semicontinuity property in \emph{one} model. 
Furthermore, \cite{nutz} gives a sufficient condition to have upper-semicontinuity, 
namely the integrability of positive parts of the utility function under \emph{every} possible model, 
see Theorem 2.2 therein and also \cite{blanchard-carassus2017} for further developments.  In our approach, upper-semicontinuity follows from the finiteness of the dual value function for \emph{every} model.}
\end{remark}

\section{Utility functions on $\mathbb{R}$}\label{sec:U_R}
 
\begin{assumption}\label{U_R}
The utility function $U :\mathbb{R} \to \mathbb{R}$  is bounded
from above, nondecreasing, concave, $U(0) = 0$. Define the convex conjugate of $U$ by
$$V(y):=\sup_{x \in \mathbb{R}}(U(x)-xy), \qquad y>0.$$  
We also assume that 
\begin{eqnarray}
\lim_{x \to -\infty } \frac{U(x)}{x} = \infty, \label{u_left_tail}\\
\limsup_{y \to \infty } \frac{V(2y)}{V(y)} < \infty. \label{v_moderate}
\end{eqnarray}
\end{assumption}

\begin{remark} {\rm Under \eqref{u_left_tail}, the function $V$ takes finite values and $V(y)>0$ for $y$ large enough, hence
 \eqref{v_moderate} makes sense.
The condition $U(0) = 0$ is used only to simplify calculations. Condition \eqref{u_left_tail} is mild
and so is \eqref{v_moderate}: as shown in Corollary 4.2(i) of \cite{s01}, for every utility function $U$ with reasonable asymptotic elasticity, its conjugate $V$ satisfies
\eqref{v_moderate}. The studies \cite{cs}, \cite{lin2017utility} assumed a smooth $U$ which is strictly
concave on its entire domain, we do not need either smoothness or strict concavity of $U$.}
\end{remark}

As discussed in \cite{bc}, \cite{s}, the choice of admissible trading strategies is a delicate issue in the context of utility maximization with utility functions that are defined on the real line. A common approach is to consider strategies whose wealth processes are bounded uniformly from below by a
constant. This choice, however, turns out to be restrictive and fails to contain optimizers. In frictionless markets, \cite{s} proved that for a utility function having reasonable asymptotic elasticity, the optimal investment process is a supermartingale under each martingale measure $Q$ such that $EV(dQ/dP)$ is finite. We will thus 
use the supermartingale property to define admissibility, just like in \cite{oz, campi}. 

To begin with, we define 
$$\mathcal{M}^{\theta}_V = \{ Q^{\theta}: (\tilde{S}^{\theta},Q^{\theta}) \text{ is a $\lambda$-consistent price system, } EV(dQ^{\theta}/dP) < \infty\},$$
the set of local martingale measures in consistent price systems for the $\theta$-model with finite generalized relative entropy.
\begin{definition}
We define 
\begin{eqnarray*}
	\mathcal{A}^{\theta}(x) &= & \left\lbrace  H \in \mathbf{V}:\ \phi_T=0,\ V^x(\theta,H)  \mbox{ is a $Q^{\theta}$-supermartingale } \right. \\
	&&\left. \mbox{ for each $\lambda$-consistent price system }  (\tilde{S}^{\theta}, Q^{\theta}) \mbox{ such that } Q^{\theta} \in \mathcal{M}^{\theta}_V \right\rbrace, 
\end{eqnarray*}
and set $\mathcal{A}(x) := \bigcap_{\theta \in \Theta} \mathcal{A}^{\theta}(x).$	
\end{definition}
 The optimization problem becomes 
\begin{equation}\label{prob_R}
u(x) = \sup_{H \in \mathcal{A}(x)} \inf_{\theta \in \Theta} EU(W^{x,liq}_T(\theta,H)).
\end{equation}

\begin{assumption}\label{weak}
	For each $\theta \in \Theta$, the price process $S^{\theta}$ admits a $\lambda$-SCPS $(Q^{\theta},\tilde{S}^{\theta})$ such that $Q^{\theta}\in
\mathcal{M}^{\theta}_V$.
\end{assumption}

\begin{remark}{\rm Unlike in \cite{cs,csy,lin2017utility} and in Section \ref{sec:U_pos}, in the present
section we do not impose the existence of consistent price systems for \emph{every} transaction cost coefficient $0 < \mu < \lambda$, we only
stipulate Assumption \ref{weak}. 
The following example shows that it is quite possible to
have CPSs for relatively large $\lambda$, without having them for
arbitrarily small $\mu$. In this example, there is an \emph{obvious arbitrage}, in the language of \cite{grw}, which persists (ceases) with sufficiently small (large) transaction costs.}
\end{remark}

\begin{example}\label{example_cost}
	{\rm Let us consider 
	$$S_t = 1 + t + \frac{1}{2\pi} \arctan(W_t), \qquad t \in [0,1].$$
	If $\lambda < 3/7$ then $(1-\lambda)S_1 > 1$ a.s, therefore, there is no consistent price system. If $\lambda\geq 2/3$, 
	then $$
	S_t(1-\lambda)\leq 3/4\leq S_t,\ t \in [0, T].
	$$
	In other words, $(\tilde{S} \equiv 3/4,P)$ is a consistent price system.}
\end{example}

\begin{theorem}\label{thm9}
Under Assumptions \ref{weak} and \ref{U_R}, there exists a strategy $H^* \in \mathcal{A}(x)$ such that
$$ u(x) = \inf_{\theta \in \Theta} EU(W^{x,liq}_T(\theta,H^*)).$$ 
\end{theorem}
\begin{proof}
We adapt certain techniques of \cite{m}. Our arguments bring novelties even in the case where $\Theta$ is a singleton (i.e. without model uncertainty). 
Define $\Phi^*(x) = -U(-x), x\ge 0.$ Its conjugate (in the sense of Subsection
\ref{madras} below) is
\begin{equation}
\Phi(y):= 
\begin{cases}
0, \text{ if } 0\le y \le \beta,\\
V(y) - V(\beta), \text{ if } y > \beta,
\end{cases} 
\end{equation}
where $\beta$ is the left derivative of $U$ at 0, see \cite{bf}. Note that 
$\Phi$, $\Phi^*$ are Young functions and $\Phi$ is of class $\Delta_2$, by (\ref{v_moderate}).

Let $H^n \in \mathcal{A}(x)$, $n\in\mathbb{N}$ be a maximizing sequence, i.e.
\begin{equation}\label{max_seq}
\inf_{\theta \in \Theta}EU(W^{x,liq}_T(\theta,H^n)) \uparrow u(x)\geq U(x).
\end{equation}
First, for all $\theta \in \Theta$, it holds that
\begin{equation}\label{U_ui} 
\sup_n EU(W^{x,liq}_T(\theta, H^n))^- < \infty.
\end{equation}
Indeed, let us assume that there exists $\theta \in \Theta$ such that (\ref{U_ui}) does not hold, or equivalently, there exists a subsequence $n_k=n^{\theta}_k$, $k\in\mathbb{N}$ such that
$EU(W^{x,liq}_T(\theta, H^{n_k}))^- > k.$ Let us denote by $C$ an upper bound of $U$, then
\begin{align*}
EU(W^{x,liq}_T(\theta, H^{n_k})) \le C - EU(W^{x,liq}_T(\theta, H^{n_k}))^- \to -\infty,
\end{align*}
$k\to\infty$, which contradicts (\ref{max_seq}). Hence (\ref{U_ui}) indeed holds. 

Consider a $\lambda$-strictly consistent price system $(\tilde{S}^{\theta},Q^{\theta})$. Fenchel's inequality gives
$$U(V^x_T(\theta,H^n)) - V((dQ^{\theta}/dP)) \le (dQ^{\theta}/dP) V^x_T(\theta,H^n)$$
and therefore $(dQ^{\theta}/dP)(V^x_T(\theta,H^n))^- \le \left( U(V^x_T(\theta,H^n)) - V((dQ^{\theta}/dP)) \right)^-.$
From (\ref{U_ui}), we deduce that 
\begin{equation}\label{bound_V_minus}
\sup_n E^{Q^{\theta}}(V^x_T(\theta,H^n))^- < \infty.
\end{equation}
It\^o's formula gives
\begin{eqnarray*}
dV^x_t(\theta,H^n) &=& - S^{\theta}_tdH^{n,\uparrow}_t + (1-\lambda)S^{\theta}_tdH^{n,\downarrow}_t + \tilde{S}^{\theta}_td\phi^n_t + \phi^n_td\tilde{S}^{\theta}_t\\
&=& (\tilde{S}^{\theta}_t - S^{\theta}_t)dH^{n,\uparrow}_t + [(1-\lambda)S^{\theta}_t - \tilde{S}^{\theta}_t]dH^{n,\downarrow}_t + \phi^n_td\tilde{S}^{\theta}_t.
\end{eqnarray*}
This implies that 
\begin{eqnarray*}
H^{n,+}_0(S^{\theta}_0-\tilde{S}_0^{\theta}) + 
\int_0^t{ (S^{\theta}_u - \tilde{S}^{\theta}_u)dH^{n,\uparrow}_u + \int_0^t 
[\tilde{S}^{\theta}_u - (1-\lambda)S^{\theta}_u ]dH^{n,\downarrow}_u}  + \left( \int_{0}^t{\phi^n_u d\tilde{S}^{\theta}_u} \right)^-   \\
\le x+ H^{n,-}_0 (S^{\theta}_0 (1-\lambda)-\tilde{S}_0^{\theta})+  (V^x_t(\theta,H^n))^- + \left( \int_0^t{\phi^n_ud\tilde{S}^{\theta}_u} \right)^+ .
\end{eqnarray*}
In particular, 
\begin{equation}\label{in_phi}
\left( \int_0^t{\phi^n_u d\tilde{S}^{\theta}_u} \right)^- \le x+ H^{n,-}_0 S^{\theta}_0(1-\lambda)  +  (V^x_t(\theta,H^n))^-.
\end{equation}
For each $n$, the process $V^x(\theta,H^n)$ is a $Q^{\theta}$-supermartingale, so there exists a 
$Q^{\theta}$-martingale which dominates the RHS of (\ref{in_phi}) and also the LHS of the same expression. Corollaire 3.5 of \cite{as}
implies that, $\int_0^t{\phi^n_u d\tilde{S}^{\theta}_u}$, $t\in [0,T]$ is a $Q^{\theta}$-supermartingale.
We get $\sup_n H^{n,-}_0<\infty$ in the same way as in the proof of Theorem \ref{thm_1}.
Consequently, (\ref{bound_V_minus}), (\ref{in_phi}) and the boundedness of $H^{n,-}_0, n \in \mathbb{N}$ give $$\sup_n E^{Q^{\theta}}\left( \int_0^T{\phi^n_t d\tilde{S}^{\theta}_t} \right)^+ < \infty.$$
Noting that $(\tilde{S}^{\theta}, Q^{\theta})$ is a $\lambda$-strictly consistent price system, we obtain from the above arguments that
\begin{eqnarray*}
\sup_n E^{Q^{\theta}}\left( H^{n,+}_0 S^{\theta}_0 +  \delta (\theta) \int_0^T[dH^{n,\uparrow}_t + dH^{n,\downarrow}_t]  \right ) & \le &  \\
\sup_n E^{Q^{\theta}}\left( H^{n,+}_0 S^{\theta}_0 +   \int_0^T{ (S^{\theta}_t - \tilde{S}^{\theta}_t)dH^{n,\uparrow}_t + [\tilde{S}^{\theta}_t - (1-\lambda)S^{\theta}_t ]dH^{n,\downarrow}  } \right) &<& \infty.
\end{eqnarray*}
Thus Lemma \ref{gom} implies the existence of convex weights $\alpha^n_j \ge 0$, $j=n,\ldots,M(n)$, 
$\sum_{j=n}^{M(n)} \alpha_j^n = 1$ such that $\tilde{H}^n:= \sum_{j=n}^{M(n)} \alpha_j^n H^n \to H^*$ in $\mathbf{V}$. Since convex combinations improve performance of concave utility functions, $\tilde{H}^n, n \in \mathbb{N}$ is also a maximizing sequence. 

We will prove that $H^* \in \mathcal{A}(x)$, in other words, the process $V^x(\theta,H^*)$ is a $Q^{\theta}$ supermartingale, for each $Q^{\theta} \in \mathcal{M}^{\theta}_V$ and for each $\theta \in \Theta$. To do so, it suffices to control the negative part of $V^x(\theta,H^*)$. It should be emphasized that (\ref{bound_V_minus}) is not enough for our purposes and a stronger statement using Orlicz space theory is needed (see Subsection \ref{madras}). 
Using concavity of $U$ and linearity of $V^x(\theta, \cdot)$, we get from (\ref{U_ui}) that
\begin{equation}\label{bound_U_tilde}
\sup_n E\left( U(V^x_T(\theta, \tilde{H}^n))\right)^- < \infty.
\end{equation}

Applying Lemma \ref{compact_orlics} to the sequence of random variables in (\ref{bound_U_tilde}), we obtain convex weights $\alpha'^n_j \ge 0$, $n\le j\le M(n)$, $\sum_{j=n}^{M(n)} \alpha'^n_j = 1$ such that  
$$Z^n:= \sum_{j=n}^{M(n)} \alpha'^n_j\left( V^x_T(\theta,\tilde{H}^n) \right)^-  $$
satisfy
\begin{equation}\label{z_ui}
L:=||\sup_nZ^n||_{\Phi^*} < \infty,
\end{equation}
By the Fenchel inequality and (\ref{z_ui}), 
\begin{equation}\label{z_Q}
E^{Q^{\theta}}\sup_n Z^n = L E^{Q^{\theta}}\left(\frac{\sup_n Z^n}{L} \right) \le L E\Phi \left( \frac{dQ^{\theta}}{dP}\right) + LE \Phi^*\left( \frac{\sup_nZ^n}{L} \right)  < \infty,
\end{equation}
for each $Q^{\theta} \in \mathcal{M}^{\theta}_V.$ Inequality (\ref{z_Q}) is trivial when $L=0$. 
Now, we define 
\begin{equation}
\overline{H}^n := \sum_{j=n}^{M(n)} \alpha'^n_j \tilde{H}^n,
\end{equation}
which is also a maximizing sequence. Using the fact that the negative part of a supermartingale is a submartingale, we get $V^x_t(\theta,\overline{H}^n)^- \le  E^{Q^{\theta}}[V^x_T(\theta,\overline{H}^n)^- |\mathcal{F}_t] $
and thus
$$\sup_n V^x_t(\theta,\overline{H}^n)^- \le \sup_nE^{Q^{\theta}}[V^x_T(\theta,\overline{H}^n)^- |\mathcal{F}_t].$$
Taking expectation both sides of the above inequality, we obtain
\begin{eqnarray*}
	E^{Q^{\theta}} \sup_n V^x_t(\theta,\overline{H}^n)^- &\le&  E^{Q^{\theta}} \left[ \sup_nE^{Q^{\theta}}[V^x_T(\theta,\overline{H}^n)^- |\mathcal{F}_t] \right] \\
	&\le & E^{Q^{\theta}} E^{Q^{\theta}}[ \sup_n V^x_T(\theta,\overline{H}^n)^- |\mathcal{F}_t]\\
	&\le& E^{Q^{\theta}} \sup_n V^x_T(\theta,\overline{H}^n)^- \\
	&\le & E^{Q^{\theta}} \sup_n Z^n   < \infty,
\end{eqnarray*}
using convexity of the mapping $x \mapsto x^-$ and (\ref{z_Q}). Since the random variable 
\begin{equation*}
\sup_{n}(V^x_t(\theta,\overline{H}^n))^-
\end{equation*}
is an upper bound of the sequence $V^x_t(\theta, \overline{H}^n)^-, n \in \mathbb{N}$,  this proves uniform integrability of that sequence under $Q^{\theta}$ at any time $t \in [0,T]$.
Also, 
\begin{equation}\label{vlacsi}
(V^x_t(\theta,{H}^*))^-\leq E^{Q^{\theta}}[ \sup_n V^x_T(\theta,\overline{H}^n)^- |\mathcal{F}_t],\ t\in [0,T],
\end{equation}
the latter process is a martingale and hence it is uniformly integrable.  

Clearly, $\overline{H}^n \to H^* \in \mathbf{V}$ and therefore 
$V^x_t(\theta,\overline{H}^n) \to V^x_t(\theta,H^*)$ a.s.  for every $t \in D\subset [0,T]$, where $[0,T]\setminus D$ has Lebesgue
measure $0$, see Remark \ref{as}. 
Let $0\le s \le t < T$ be both in $D$. Noting the supermartingale property, Fatou's lemma yields
\begin{eqnarray*}
E^{Q^{\theta}}\left[ V^x_t(\theta,H^*)|\mathcal{F}_s\right]  &=& E^{Q^{\theta}}\left[ \liminf_{n \to \infty}V^x_t(\theta,\overline{H}^n)|\mathcal{F}_s\right] \le \liminf_{n\to \infty} E^{Q^{\theta}}\left[ V^x_t(\theta,\overline{H}^n) |\mathcal{F}_s\right]\\
&\leq& \liminf_{n\to\infty} V^x_s(\theta,\overline{H}^n) = V^x_s(\theta,H^*).
\end{eqnarray*} 
The same argument works for $t=T$, too.
Now it extends to arbitrary $t\in [0,T]$ using Fatou's lemma and \eqref{vlacsi}. Finally, it extends to arbitrary
$s\in [0,T]$ by the backward martingale convergence theorem and by
right-continuity of $t\to V^x_t(\theta,H^*)$. This means that $V^x(\theta,H^*)$ is a $Q^{\theta}$-supermatingale and therefore $H^* \in \mathcal{A}(x)$.

Since $U$ is bounded from above, by Fatou's lemma 
\begin{align*}
\limsup_{n\to\infty}\left(\inf_{\theta \in \Theta} EU(W^{x,liq}_T(\theta, \overline{H}^n)) \right) &\le \inf_{\theta \in \Theta} \limsup_{n\to\infty} EU(W^{x,liq}_T(\theta, \overline{H}^n))\\
&\le  \inf_{\theta \in \Theta} E U(W^{x,liq}_T(\theta, H^*)),  
\end{align*} 
which proves the optimality of $H^*$. 
\end{proof}

\section{Conclusions}\label{qclu}

It is possible to extend our results in Section \ref{sec:U_R}:
one could treat the multi-asset ``conic'' framework of \cite{ks}; unbounded utilities could also be incorporated along the lines of Theorem 3.12 in
\cite{m}; random endowments (or random utilities) can be added at little
cost since we do not consider the dual problem at all. These extensions, however, require no essential new ideas
while they would considerably complicate the presentation. 
Our emphasis here is on introducing a new approach, and not on striving
for the utmost generality. 

Admitting jumps in the price process leads to a 
more involved class of strategies. The treatment of that setting is a direction of research worth
pursuing in the future.
  
\section{Appendix}\label{appendix}

\subsection{Finite variation processes}\label{section:v_space}

Let $\mathcal{V}$ denote the family of non-decreasing, right-continuous functions on $[0,T]$ which are $0$ at $0$.
Let $r_k$, $k\in\mathbb{N}$ be an enumeration of $D:=\left(\mathbb{Q}\cap [0,T]\right)\cup \{T\}$ with $r_0=T$.
For $f,g\in\mathcal{V}$, define
\[
\rho(f,g):=\sum_{k=0}^{\infty} 2^{-k}|f(r_k)-g(r_k)|.
\]
The series converges since $|f(r_k)-g(r_k)|\leq f(T)+g(T)$, and it defines a metric. The corresponding Borel-field
is denoted by $\mathcal{G}$.

Let $\mathbf{V}$ denote the set of triplets 
$H=(H^{\uparrow},H^{\downarrow},H_0)$ where
$H^{\uparrow}_t,H^{\downarrow}_t$, $t\in [0,T]$ are optional processes 
such that $H^{\uparrow}(\omega),H^{\downarrow}(\omega)\in\mathcal{V}$ for each 
$\omega\in\Omega$ and $H_0\in\mathbb{R}$ (deterministic).
Considered as mappings $H^{\uparrow},H^{\downarrow}:(\Omega, \mathcal{F}) \to (\mathcal{V}, \mathcal{G})$, they are measurable, 
by the definition of the metric $\rho$.
We identify elements of $\mathbf{V}$ when they coincide outside a $P$-null set. 
We say that a sequence $H^n\in\mathbf{V}$ is convergent to some $H\in\mathbf{V}$ if $H^{n,\uparrow}\to H^{\uparrow}$ and
$H^{n,\downarrow}\to H^{\downarrow}$ a.s. in $\mathcal{V}$, $n\to\infty$ and also $H^n_0\to H_0$ (in the topology of $\mathbb{R}$).


Convex compactness-type results for finite variation processes have been introduced in various forms in the
literature. The next result is very similar to Lemma 3.5 in \cite{ka}.

\begin{lemma}\label{gom} Let 
$H^n\in\mathbf{V}$, $n\in\mathbb{N}$ 
be such that $$
\sup_{n\in\mathbb{N}}\left(E^Q[H^{n,\uparrow}_T+H^{n,\downarrow}_T] + |H^n_0|\right)<\infty
$$
for some $Q\sim P$.
	Then there is $H\in\mathbf{V}$ and there
	are convex weights $\alpha^n_j\geq 0$, $j=n,\ldots,M(n)$, $\sum_{j=n}^{M(n)}\alpha^n_j=1$, $n\in\mathbb{N}$
	such that 
	\[
	\tilde{H}^n:=\sum_{j=n}^{M(n)}\alpha_j^n H^j\to H\mbox{ in }\mathbf{V}.
	\]
It follows that, for $P$-almost every $\omega\in\Omega$,
$$
\tilde{H}^{n,\uparrow}_t(\omega) \to H^{\uparrow}_t(\omega)\mbox{ and }
\tilde{H}^{n,\downarrow}_t(\omega)\to H^{\downarrow}_t(\omega)
$$
at $t=T$ and each $t$ which is a continuity point of both $H^{\uparrow}(\omega)$ and $H^{\downarrow}(\omega)$. 
\end{lemma}
\begin{proof}
Recall that $D = ([0,T] \cap \mathbb{Q}) \cup \{T\}$. By assumption, the sequence $H^{n,\uparrow}_T, n \in \mathbb{N}$ is bounded in $L^1(Q)$ for some $Q \sim P$ so we use the Koml\'os theorem together with a diagonalization procedure to obtain sequences of convex weights $\alpha^n_j$ such that
\begin{equation}\label{eq_conv_rat}
\tilde{H}^{n,\uparrow}_t \to H^{\uparrow}_t, \qquad t \in D
\end{equation}
for some $\mathcal{F}_t$-measurable random variables $H^{\uparrow}_t$, on an event $\tilde{\Omega}$ with $P[\tilde{\Omega}] =  1$. Since the limiting process, if exists, is non-decreasing and right-continuous, we set
$$H^{\uparrow}_t = \lim_{q \downarrow t, q \in \mathbb{Q}} H^{\uparrow}_q, \qquad t \in [0,T).$$
We prove that, for $\omega\in\tilde{\Omega}$,
\begin{equation}\label{eq_conv_all_t}
\tilde{H}^{n,\uparrow}_t(\omega) \to H^{\uparrow}_t(\omega), 
\end{equation}
for each $t \in [0,T)$ that is a continuity point of the function $s \to H^{\uparrow}_s(\omega)$.
Fix $\varepsilon > 0$ arbitrarily. Using continuity at $t$ of $H^{\uparrow}$, we find two rational numbers $q_1, q_2$ such that $q_1 < t < q_2$ and that $H^{\uparrow}_{q_2}(\omega) - H^{\uparrow}_{q_1}(\omega) < \varepsilon.$ From (\ref{eq_conv_rat}), there exists $N = N(\omega)$ such that 
$$|\tilde{H}^{n,\uparrow}_{q_2}(\omega) - H^{\uparrow}_{q_2}(\omega)| < \varepsilon, \qquad |\tilde{H}^{n,\uparrow}_{q_1}(\omega) - H^{\uparrow}_{q_1}(\omega)| < \varepsilon, \qquad \forall n \ge N.$$ 
We estimate for all $n \ge N$
\begin{eqnarray*}
|\tilde{H}^{n,\uparrow}_{q_2}(\omega) - \tilde{H}^{n,\uparrow}_{q_1}(\omega)| &\le& |\tilde{H}^{n,\uparrow}_{q_2}(\omega) - H^{\uparrow}_{q_2}(\omega)| + |H^{\uparrow}_{q_2}(\omega) - H^{\uparrow}_{q_1}(\omega)| \\
&+&  |H^{\uparrow}_{q_1}(\omega) - \tilde{H}^{n,\uparrow}_{q_1}(\omega)|\le 3\varepsilon.
\end{eqnarray*}
Therefore, using monotonicity of $\tilde{H}^{n,\uparrow}$, we obtain for all $n \ge N(\omega)$
\begin{eqnarray*}
|\tilde{H}^{n,\uparrow}_t(\omega) - H^{\uparrow}_t(\omega)| &\le& |\tilde{H}^{n,\uparrow}_t(\omega) - \tilde{H}^{n,\uparrow}_{q_2}(\omega)| + |\tilde{H}^{n,\uparrow}_{q_2}(\omega) - H^{\uparrow}_{q_2}(\omega)| + |H^{\uparrow}_{q_2}(\omega) - H^{\uparrow}_t(\omega)|\\
&\le& |\tilde{H}^{n,\uparrow}_{q_1}(\omega) - \tilde{H}^{n,\uparrow}_{q_2}(\omega)| + |\tilde{H}^{n,\uparrow}_{q_2}(\omega) - H^{\uparrow}_{q_2}(\omega)| + |H^{\uparrow}_{q_2}(\omega) - H^{\uparrow}_{q_1}(\omega)|\\
&\le&  5 \varepsilon.
\end{eqnarray*}
Notice that \eqref{eq_conv_all_t} also holds for $t=T$.
The same argument can be repeated for the sequence $\tilde{H}^{n,\downarrow}, n \in \mathbb{N}$ and also $H_0^n\to H_0$ can
be guaranteed with some $H_0\in\mathbb{R}$ by extracting a further subsequence.

\end{proof}

\begin{remark}\label{as}
	{\rm The above proof shows that if $f_n\to f$, $n\to\infty$ in $\mathcal{V}$ then $f_n(x)$ tends to $f(x)$ in every continuity point $x$ of $f$. Consequently, for any continuous  $g:[0,T]\to\mathbb{R}$, $\int_0^T g(t)\,df_n(t)\to \int_0^T g(t)\,df(t)$, $n\to\infty$ where integration is meant with respect to the measures induced by $f_n$, $f$.

Hence for the sequence $\tilde{H}^n$ constructed in Lemma
\ref{gom} above, $W^{x,liq}_t(\theta,\tilde{H}^n)(\omega)\to
W^{x,liq}_t(\theta,H)(\omega)$ and $V^{x}_t(\theta,\tilde{H}^n)(\omega)\to
V^{x}_t(\theta,H)(\omega)$, $n\to\infty$ almost surely in $t=T$ and in every $t$ which is a continuity point of both $H^{\uparrow}(\omega), H^{\downarrow}(\omega)$, in particular, for Lebesgue-a.e.\ $t$.
Fubini's theorem
thus implies that there is a set $Z$ of zero Lebesgue-measure (excluding $T$) such that for $t\in [0,T]\setminus Z$,
$W^{x,liq}_t(\theta,\tilde{H}^n)\to
W^{x,liq}_t(\theta,H)$ and $V^{x}_t(\theta,\tilde{H}^n)\to
V^{x}_t(\theta,H)$ hold $P$-almost surely.}
\end{remark}

\subsection{Orlicz spaces}\label{madras}

We call $\Phi:\mathbb{R}_+ \to \mathbb{R}_+$ a \emph{Young function} if it is convex with
$\Phi(0) = 0$ and $\lim_{x\to \infty}\Phi(x)/x = \infty$. The set
$$L^{\Phi}:= \{ X \in L^0: E\Phi(\gamma |X|) < \infty \text{ for some } \gamma > 0 \} $$
is a Banach space with the following norm
$$ ||X||_{\Phi}:= \inf\{ \gamma > 0: X \in \gamma B_{\Phi} \} $$
where $B_{\Phi}:=\{ X \in L^0: E\Phi(|X|) \le 1 \}$, the unit ball of $L^{\Phi}$. Define the conjugate function $\Phi^*(y):=\sup_{x\ge 0}(xy - \Phi(x)), y \in \mathbb{R}_+$. This is also a Young function and $(\Phi^*)^* = \Phi.$
We say that $\Phi$ is of class $\Delta_2$ if 
$$ \limsup_{x\to \infty} \frac{\Phi(2x)}{\Phi(x)} < \infty.$$
We recall Corollary 3.10 of \cite{do}, a compactness result which will be used to handle the losses of trading strategies in this paper.
\begin{lemma}\label{compact_orlics}
Let $\Phi$ be a Young function of class $\Delta_2$ and let $\xi_n, n\ge 1$ be a norm-bounded sequence in $L^{\Phi^*}$. Then there are convex weights $\alpha^n_j \ge 0$, $n \le j \le M(n)$, $ \sum_{j=n}^{M(n)}\alpha^n_j = 1$ such that
$$\xi'_n:= \sum_{j=n}^{M(n)} \alpha^n_j \xi_j $$
converges almost surely to some $\xi \in L^{\Phi^*}$ as $n\to\infty$, and $\sup_{n}|\xi'_n|$ is in $L^{\Phi^*}$.\hfill $\Box$
\end{lemma}

\bibliography{robust}{}
\bibliographystyle{plain}

\end{document}